\documentclass[journal]{IEEEtran}
\usepackage{cite}

\ifCLASSINFOpdf
\else
\fi
\usepackage{amssymb}
\usepackage{amsthm}
\usepackage[cmex10]{amsmath}

\usepackage[colorlinks=true,linkcolor=black,citecolor=black,plainpages=false,pdfpagelabels]{hyperref}
\hypersetup{pdftitle={Chain rules for smooth min- and max-entropies}}

\newcommand{\eps}{\varepsilon}
\newcommand{\Hil}{\mathcal{H}}
\newcommand{\Po}{\mathcal{P}}
\newcommand{\Su}{\mathcal{S}}

\newtheorem{thm}{Theorem}
\newtheorem{lemma}[thm]{Lemma}
\newtheorem*{lemma*}{Lemma}
\newtheorem{cor}[thm]{Corollary}
\newtheorem{definition}[thm]{Definition}

\DeclareMathOperator{\Herm}{Herm}
\DeclareMathOperator{\rank}{rank}

\DeclareMathOperator{\tr}{tr}
\DeclareMathOperator{\supp}{supp}
\DeclareMathOperator{\id}{\mathbb{I}}

\newtheoremstyle{named}{}{}{\itshape}{}{\bfseries}{.}{.5em}{\thmnote{#3}}
\theoremstyle{named}
\newtheorem*{namedlemma}{Lemma}

\usepackage[normalem]{ulem}

\begin{document}
%
\title{Chain Rules for Smooth Min- and Max-Entropies}
%
%
%

\author{Alexander~Vitanov,
				Fr\'ed\'eric~Dupuis,
        Marco~Tomamichel,
        and~Renato~Renner
\thanks{A.~Vitanov is with the Department of Mathematics, ETH Zurich, R\"amistrasse 101, 8092 Z\"urich. This work was produced while F.~Dupuis was with the Institute for Theoretical Physics, ETH Zurich, 8093 Z\"urich, Switzerland. Since January 2012 he is with the Department of Computer Science at University of Aarhus, {\AA}bogade 34, 8200 Aarhus N, Denmark. R.~Renner is with the Institute for Theoretical Physics, ETH Zurich, 8093 Z\"urich, Switzerland. M.~Tomamichel is with the Center for Quantum Technologies, National University of Singapore, 3 Science Drive 2, Singapore 117543. (e-mail: alexander.vitanov@math.ethz.ch; dupuis@cs.au.dk; cqtmarco@nus.edu.sg; renner@phys.ethz.ch)

Copyright (c) 2012 IEEE. Personal use of this material is permitted.  However, permission to use this material for any other purposes must be obtained from the IEEE by sending a request to pubs-permissions@ieee.org.}}

\maketitle

\begin{abstract}
  The chain rule for the Shannon and von Neumann entropy, which
  relates the total entropy of a system to the entropies of its parts,
  is of central importance to information theory. Here we consider the
  chain rule for the more general smooth min- and max-entropy, used
  in one-shot information theory. For these entropy measures, the
  chain rule no longer holds as an equality. However,
  the standard chain rule for the von Neumann entropy is retrieved asymptotically 
  when evaluating them for many identical and independently distributed states.

\end{abstract}   
\IEEEpeerreviewmaketitle

\section{Introduction}
%
%
%
%

\IEEEPARstart{I}{n} classical and quantum information theory, entropy
measures are often used to characterize fundamental information
processing tasks. For example, in his groundbreaking work on
information and communication theory~\cite{shannon48}, Shannon showed
that entropies can be used to quantify the memory needed to store the
(compressed) output of an information source or the capacity of a
communication channel.  It follows immediately from the basic
properties of the Shannon entropy that the equality
  \begin{align*}
    H(AB) = H(A|B) + H(B) \ ,
  \end{align*}
  which we call the \emph{chain rule}, must hold. Here, $H(B)$ denotes
  the entropy of the random variable $B$ and $H(A|B)$ is the entropy
  of the random variable $A$ averaged over \emph{side information} in
  $B$. The chain rule therefore asserts that the entropy of two
  (possibly correlated) random variables, $A$ and $B$, can be
  decomposed into the entropy of $B$ alone plus the entropy of $A$
  conditioned on knowing $B$. More generally, one may average over
  additional side information, $C$, in which case the chain rule takes
  the more general form
  \begin{align} \label{eq_chain}
   H(AB|C) =  H(A|BC) + H(B|C) \ .
  \end{align} 
  The chain rule forms an integral part of the entropy calculus.  The
  other basic ingredient is strong sub-additivity, which can be
  written as $H(A|BC) \leq H(A|C)$, i.e.\ additional side information
  can only decrease the entropy.

  The quantum generalization of Shannon's entropy, the \emph{von
    Neumann entropy}, inherits these fundamental properties. For a
  quantum state\footnote{Formal definitions follow in
    Section~\ref{sec:math-preliminaries}.} $\rho_A$ on $A$, the von
  Neumann entropy is defined as $H(A)_{\rho} := - \tr ( \rho_A \log
  \rho_A)$, where $\tr$ denotes the trace and $\log$ is taken in base
  2 throughout this paper. The conditional von Neumann entropy with
  classical side information can again be defined by an average,
  however, this intuitive definition fails if the side information is
  quantum.  Pointing to its fundamental importance, the conditional
  von Neumann entropy is thus defined by the chain rule itself, i.e.\
  $H(A|B)_{\rho} := H(AB)_{\rho} - H(B)_{\rho}$. In addition to the
  chain rule and strong sub-additivity, it also satisfies a duality
  relation: For any pure tripartite state $\rho_{ABC}$, we have
  $H(A|B)_{\rho} = -H(A|C)_{\rho}$.

  Shannon and von Neumann entropies have been successfully employed to
  characterize an enormous variety of information theoretic tasks,
  many of which are of high practical relevance (examples include the
  aforementioned tasks of data compression or channel
  coding). However, a basic assumption usually made in this context is
  that the underlying random processes (e.g., those relevant for the
  generation of data, or the occurrence of noise in a communication
  channel) are modeled asymptotically by an arbitrarily long sequence
  of random variables that are \emph{independent and identically
    distributed (i.i.d.).}  In the absence of this assumption (e.g.,
  if a channel is only invoked a small number of times or if its noise
  model is not i.i.d.), the use of the von Neumann entropy is
  generally no longer justified. The formalism of smooth min- and
  max-entropy, introduced in~\cite{RW04, Ren05, RK05} and further
  developed in~\cite{KRS09,TCR09, TCR10,Dat09}, overcomes this
  limitation and enables the analysis of general situations beyond the
  i.i.d.\ scenario. This level of generality turned out to be crucial
  in various areas, e.g., in physics (where entropies are employed for
  the analysis of problems in thermodynamics~\cite{RARDV11}) or in
  cryptography (where entropies are used to quantify an adversary's
  uncertainty).

  Smooth min- and max-entropy, denoted $H_{\min}^{\eps}$ and
  $H_{\max}^{\eps}$, respectively, depend on a positive real value
  $\eps$, called \emph{smoothing parameter} $\eps$ (see
  Section~\ref{sec:math-preliminaries} for formal definitions). When
  the entropies are used to characterize operational tasks, the
  smoothing parameter determines the desired accuracy. For example,
  the smooth min-entropy, $H_{\min}^{\eps}(A|B)$, characterizes the
  number of fully mixed qubits, independent (i.e.\ \emph{decoupled})
  from side information $B$, that can be extracted from a quantum
  source $A$~\cite{dupuis09,dupuis10}. Furthermore, the smooth
  max-entropy, $H_{\max}^{\eps}(A|B)$, characterizes the amount of
  entanglement needed between two parties, $A$ and $B$, to merge a
  state $\rho_{AB}$, where $\rho_A$ is initially held by $A$, to
  $B$~\cite{berta08,dupuis10}. In both cases, the smoothing parameter
  $\eps$ corresponds to the maximum distance between the desired final
  state and the one that can be achieved.

  Smooth entropy can be seen as strict generalization of Shannon or
  von Neumann entropy. In particular, the latter can be recovered by
  evaluating the smooth min- or max-entropy for i.i.d.\
  states~\cite{Ren05,TCR09}. Accordingly, smooth entropy inherits many
  of the basic features of von Neumann entropy, such as strong
  sub-additivity. In light of this, it should not come as a surprise
  that smooth entropy also obeys inequalities that generalize the
  chain rule~\eqref{eq_chain}. Deriving these is the main aim of this
  work.

  Specifically, one can obtain four pairs of generalized chain
  inequalities. For any small smoothing parameters $\eps', \eps'',\eps''' 
  \geq 0$ and
  $\eps > \eps' + 2\eps''$, we have
\begin{align*}
H_{\min}^{\eps}(AB|C)_{\rho} &\geq H_{\min}^{\eps''}(A|BC)_{\rho}+H_{\min}^{\eps'}(B|C)_{\rho} -f \,, \\
H_{\max}^{\eps}(AB|C)_{\rho} &\leq H_{\max}^{\eps'}(A|BC)_{\rho}+H_{\max}^{\eps''}(B|C)_{\rho} +f \,,
\end{align*}
\vspace{-0.9cm}

\begin{align*}
H_{\min}^{\eps'}(AB|C)_{\rho} &\leq H_{\min}^{\eps}(A|BC)_{\rho}+H_{\max}^{\eps''}(B|C)_{\rho}
+2f  \,,\\
H_{\max}^{\eps'}(AB|C)_{\rho} &\geq H_{\min}^{\eps''}(A|BC)_{\rho}+H_{\max}^{\eps}(B|C)_{\rho}
-2f \,,
\end{align*}
\vspace{-0.9cm}

\begin{align*}
H_{\min}^{\eps'}(AB|C)_{\rho} &\leq H_{\max}^{\eps''}(A|BC)_{\rho}+ H_{\min}^{\eps}(B|C)_{\rho}
+3f  \,,\\
H_{\max}^{\eps'}(AB|C)_{\rho} &\geq H_{\max}^{\eps}(A|BC)_{\rho}+ H_{\min}^{\eps''}(B|C)_{\rho}
-3f  \,,
\end{align*}
\vspace{-0.9cm}

\begin{align*}
H_{\min}^{\eps'}(AB|C)_{\rho} &\leq
H_{\max}^{\eps'''}(A|BC)_{\rho}+H_{\max}^{\eps''}(B|C)_{\rho}+g  \,,\\
H_{\max}^{\eps'}(AB|C)_{\rho}&\geq
H_{\min}^{\eps''}(A|BC)_{\rho}+H_{\min}^{\eps'''}(B|C)_{\rho}-g  \,,
\end{align*}
%
where $f$ does not grow more than of the order $\log 1/e$ when $e = \eps - \eps' - 2\eps''$ is small, 
and $g$ is smaller than $6$ for $\eps' + 2\eps'' + \eps''' < 1/5$.
We note that, in typical applications, we would choose the smoothing
parameters so that the correction terms $f$ and $g$ are small compared
to the typical values of the smooth entropies.

The fact that generalized chain inequalities hold for smooth min- and
max-entropy is not only important for establishing a complete entropy
calculus, analogous to that for the von Neumann entropy. They are also
crucial for applications, as the following example shows.

In quantum key distribution, after the quantum signals have been exchanged and measured, two honest parties, Alice and Bob, are left with two correlated raw keys, about which a potential eavesdropper is guaranteed to have only limited information. This limit on the eavesdropper's knowledge is best expressed~\cite{Ren05} by a bound on the smooth min-entropy of Alice's raw key, $X_A$, conditioned on the eavesdropper's quantum information, $E$, i.e., $H_{\min}^{\eps'}(X_A|E)$. However, to ensure that Bob's final key agrees with her own, Alice will have to send a syndrome, $S = s(X_A)$, over an insecure channel. A fundamental question in quantum key distribution is thus to bound
$H_{\min}^{\eps}(X_A|E S)$, i.e., the smooth min-entropy of $X_A$ conditioned on the eavesdropper's information after learning $S$. The third chain rule above states that
\begin{align*}
  H_{\min}^{\eps}(X_A|E S) &\geq H_{\min}^{\eps'}(X_A S | E) - H_{\max}^{\eps''}(S|E) - 2f \\
  &= H_{\min}^{\eps'}(X_A | E) - H_{\max}^{\eps''}(S|E) - 2f .
\end{align*}
Here, we used that $S = s(X_A)$ and thus $X_A \to X_A S$ is an isometry under which the smooth entropies are invariant~\cite{TCR10}. Roughly speaking, our chain rule thus implies that the eavesdropper gains at most $H_{\max}^{\eps''}(S|E)$ bits of information about $X_A$, where we assumed that $f$ is negligible. This is strictly tighter than previous results (see, e.g.,~\cite{WTHR11}), where the gain was bounded by $\log |S| \geq H_{\max}^{\eps''}(S|E)$, where $|S|$ is the number of different syndromes that can be stored in $S$.
This leads to strictly tighter bounds, for instance, when $S$ contains information that has been communicated previously over the public channel and is therefore already included in $E$.

Until now, only special cases of
the above inequalities have been known, except for the first pair, which
has been derived in~\cite{dupuis10}. In the present paper we provide
proofs for the remaining relations.  In fact, since smooth min- and
max-entropy obey a duality relation similar to that of von Neumann
entropy~\cite{TCR10}, $H_{\min}^{\eps}(A|B) = - H_{\max}^{\eps}(A|C)$, 
the paired inequalities above imply each
other. It will therefore suffice to prove only one inequality of each
pair.

The paper is organized as follows. In the next section we introduce
the notation, terminology, and basic definitions. In particular, we
define the (smooth) min- and max-entropy measures and outline some of
their basic features. In Section~\ref{sec:hmax-expressions} we derive
alternative expressions for the max-entropy based on semidefinite
programming duality. While these expressions may be of independent
interest, they will be used in Section~\ref{sec:main-results}, which
is devoted to the statement and proofs of the generalized chain rules.

\section{Mathematical Preliminaries}
\label{sec:math-preliminaries}
\subsection{Notation and basic definitions}
Throughout this paper we focus on finite dimensional Hilbert spaces. Hilbert spaces corresponding to different physical systems are distinguished by different capital Latin letters as subscript $\Hil_A, \Hil_B$ etc. The tensor product of $\Hil_A$ and $\Hil_B$ is designated in short by $\Hil_{AB}=\Hil_A\otimes\Hil_B$.

The set of linear operators from $\Hil_A$ to $\Hil_B$ is denoted by $\mathcal{L}(\Hil_A, \Hil_B)$. The space of linear operators acting on the Hilbert space $\Hil$ is denoted by $\mathcal{L}(\Hil)$ and the subset of $\mathcal{L}(\Hil)$ containing the Hermitian operators on $\Hil$ is denoted by $\Herm(\Hil)$. Note that $\Herm(\Hil)$ endowed with the Hilbert-Schmidt inner product $\left\langle X,Y\right\rangle := \tr(X^{\dagger}Y)$, $X, Y\in\Herm(\Hil)$, is a Hilbert space. Given an operator $R\in\Herm(\Hil)$, we write $R\geq0$ if and only if $R$ is positive semi-definite and $R>0$ if and only if it is positive definite. Furthermore, let $\Su_{\leq}(\Hil)$ and $\Su_{=}(\Hil)$ denote the sets of sub-normalized and normalized positive semi-definite \emph{density operators} with $\tr\rho\leq1$ and $\tr\rho=1$, respectively.

Inequalities between Hermitian operators are defined in the following sense: Let $R,S\in\Herm(\Hil)$, then we write $R\geq S$, respectively $R>S$ if and only if $R-S$ is positive semi-definite, respectively positive definite.

Given an operator $R$, the operator norm of $R$ is denoted by $\|R\|_{\infty}$ and is equal to the highest singular value of $R$. The trace norm of $R$ is given by $\|R\|_{1}:=\tr[\sqrt{R^{\dagger}R}]$. The fidelity between two states $\rho$, $\sigma\in\Su_{\leq}(\Hil)$ is defined as $F(\rho, \sigma):=\|\sqrt{\rho}\sqrt{\sigma}\|_{1}$.



For multipartite operators on product spaces $\Hil_{AB}$ we will use subscripts to denote the space on which they act (e.g. $S_{AB}$ for an operator on $\Hil_{AB}$). Given a multipartite operator $S_{AB}\in\mathcal{L}(\Hil_{AB})$, the corresponding reduced operator on $\Hil_A$ is defined by $S_A:=\tr_{B}[S_{AB}]$ where $\tr_B$ denotes the partial trace operator on the subsystem $\Hil_B$. Given a multipartite operator $S_{AB}$ and the corresponding marginal operator $S_A$, we call $S_{AB}$ an \emph{extension} of $S_A$. We omit identities from expressions which involve multipartite operators whenever mathematically meaningful expressions can be obtained by tensoring the corresponding identities to the operators.

\subsection{Smooth Min- and Max-Entropies}
In the following we successively give the definitions of the non-smooth min- and max-entropies and their smooth versions \cite{Ren05}, \cite{KRS09}.

\begin{definition}
\label{minentropy}
Let $\rho_{AB}\in\Su_{\leq}(\Hil_{AB})$, then the min-entropy of $A$ conditioned on $B$ of $\rho_{AB}$ is defined as
\begin{align}
H_{\min}(A|B)_{\rho} &:=\max_{\sigma_{B}\in\Su_{\leq}(\Hil_B)}H_{\min}(A|B)_{\rho|\sigma}, \quad \textrm{where} \nonumber\\
H_{\min}(A|B)_{\rho|\sigma} &:=\sup\bigl\{\lambda\in\mathbb{R}: \rho_{AB}\leq2^{-\lambda}\id_A\otimes\sigma_{B}\bigr\}.
\label{eqminentropy}
\end{align}
\end{definition}
Note that $H_{\min}(A|B)_{\rho|\sigma}$ is finite if and only if $\supp(\rho_{B})\subseteq\supp(\sigma_B)$ and divergent otherwise.  

\begin{definition}
\label{maxentropy}
Let $\rho_{AB}\in\Su_{\leq}(\Hil_{AB})$, then the max-entropy of A conditioned on B of $\rho_{AB}$ is defined as
\begin{align}
H_{\max}(A|B)_{\rho} &:=\max_{\sigma_B\in\Su_{\leq}(\Hil_B)}H_{\max}(A|B)_{\rho|\sigma}, \quad \textrm{where}\nonumber\\
\label{eqmaxentropy}
H_{\max}(A|B)_{\rho|\sigma} &:=\log F(\rho_{AB}, \id_A\otimes\sigma_B)^{2}.
\end{align}
\end{definition}
The maximum in \eqref{eqminentropy} and \eqref{eqmaxentropy} is achieved at $\Su_{=}(\Hil_{B})$. The $\eps$-smooth min- and max-entropies of a state $\rho$ can be understood as an optimization of the corresponding non-smooth quantities over a set of states $\eps$-close to $\rho$. We use the \emph{purified distance} to quantify the $\eps$-closeness of states.
\begin{definition}
Let $\rho$, $\sigma\in\Su_{\leq}(\Hil)$. Then the purified distance between $\rho$ and $\sigma$ is defined by
\begin{equation}
P(\rho, \sigma):=\sqrt{1-\bar{F}(\rho, \sigma)^{2}}, \quad \textrm{where}
\end{equation}
\begin{equation}
\bar{F}(\rho, \sigma):= F(\rho, \sigma)+\sqrt{\bigl(1-\tr\rho\bigr)\bigl(1-\tr\sigma\bigr)}
\end{equation}
is the generalized fidelity.
\end{definition}

Hereafter, when two states $\rho, \sigma\in\Su_{\leq}(\Hil)$ are
said to be $\eps$-close we mean $P(\rho, \sigma)\leq\eps$ and denote
this by $\rho\approx_{\eps}\sigma$. Some of the basic properties of
the purified distance are reviewed in Appendix
\ref{sec:appendix-tech-lemmas}, but for a more comprehensive treatment
we refer to \cite{TCR10}. With that convention we are ready to
introduce a smoothed version of the min- and max-entropies
\cite{Ren05}.
\begin{definition}
\label{smoothmin}
Let $\eps\geq0$, $\rho_{AB}\in\Su_{\leq}(\Hil_{AB})$. Then the $\eps$-smooth min-entropy of A conditioned on B of $\rho_{AB}$ is defined as
\begin{equation}
H_{\min}^{\eps}(A|B)_{\rho}:=\max_{\tilde{\rho}} H_{\min}(A|B)_{\tilde{\rho}}
\end{equation}
and the $\eps$-smooth max-entropy of A conditioned on B of $\rho_{AB}$ is defined as
\begin{equation}
H_{\max}^{\eps}(A|B)_{\rho}:=\min_{\tilde{\rho}} H_{\max}(A|B)_{\tilde{\rho}}
\end{equation}
where the maximum and the minimum range over all sub-normalized states $\tilde{\rho}_{AB}\approx_{\eps}\rho_{AB}$.
\end{definition}
The smooth entropies are dual to each other in the following sense.
When $\rho_{ABC}\in\Su_{\leq}(\Hil_{ABC})$ is pure, we have~\cite{TCR10} 
\begin{equation}
\label{smoothdual}
H_{\max}^{\eps}(A|B)_{\rho}=-H_{\min}^{\eps}(A|C)_{\rho}.
\end{equation}  

Finally, the smooth min-entropy is upper-bounded by the smooth
max-entropy as shown by the following lemma whose proof is deferred to
Appendix \ref{sec:various-proofs}:
\begin{lemma}
\label{updown}
Let $\eps$, $\eps^{\prime}\geq 0$ and let $\rho_{AB}\in\Su_{\leq}(\Hil_{AB})$ be such that $\eps+\eps^{\prime}+2\sqrt{1-\tr\rho_{AB}}<1$. Then,
\begin{equation}
\label{minmax}
\begin{split}
H_{\min}^{\eps^{\prime}}(A|B)_{\rho}&\leq H_{\max}^{\eps}(A|B)_{\rho}\\
&+\log\left(\frac{1}{1-(\eps+\eps^{\prime}+2\sqrt{1-\tr\rho})^{2}}\right).
\end{split}
\end{equation}
\end{lemma}
\subsection{Semidefinite Programming}
This subsection is devoted to the duality theory of semi-definite programs (SDPs). We will present the subject as given in \cite{Bar} and especially in \cite{JW} but will restrict the discussion to the special case which is of interest in this work.\\ 
A semidefinite program over the Hilbert spaces $\Hil_A$ and $\Hil_B$ is a triple $(\mathcal{F}, R_A, S_B)$, $\mathcal{F}\in\mathcal{L}(\Herm(\Hil_A), \Herm(\Hil_B))$, $R_A\in\Herm(\Hil_A)$ and $S_B\in\Herm(\Hil_B)$, which is associated with the following two optimization problems:

\vspace{5mm}
\small
\begin{minipage}[b]{0.45\linewidth}\centering
\begin{tabular}{cc}
\hspace{-0.0cm}&\hspace{-1cm}\textsc{Primal Problem:}\\
&\\
\hspace{-0.1cm}minimize:&\hspace{-0.5cm}$\tr[R_AX_A]$\\
\hspace{-0.1cm}subject to:&\hspace{-0.2cm}$\mathcal{F}(X_A)\geq S_B$\\
\hspace{-0.1cm}&$X_A\geq0$
\end{tabular}
\end{minipage}
\hspace{-0.3cm}
\begin{minipage}[b]{0.45\linewidth}
\centering
\begin{tabular}{cc}
\hspace{-0.0cm}&\hspace{-1cm}\textsc{Dual problem}:\\
&\\
\hspace{-0.0cm}maximize:&\hspace{-0.5cm}$\tr[S_BY_B]$\\
\hspace{-0.0cm}subject to:&\hspace{-0.2cm}$\mathcal{F}^{\dagger}(Y_B)\leq R_A$\\
\hspace{-0.0cm}&$Y_B\geq0$
\end{tabular}
\end{minipage}
\normalsize
\vspace{5mm}

where $X_A\in\Herm(\Hil_A)$ and $Y_B\in\Herm(\Hil_B)$ are variables. $X_A\geq0$ and $Y_B\geq0$ such that $\mathcal{F}(X_A)\geq S_B$ and $\mathcal{F}^{\dagger}(Y_B)\leq R_A$, respectively, are called \emph{primal feasible plan} and \emph{dual feasible plan}, respectively. We also denote the solutions to the primal and dual problems by
\[\gamma:=\inf\bigl\{\tr[R_AX_A]: X_A~\text{is a primal feasible plan}\bigl\},\]
\[\delta:=\sup\bigl\{\tr[S_BY_B]: Y_B~\text{is a dual feasible plan}\bigr\}.\]
The values $X_A\geq0$ and $Y_B\geq0$ satisfying $\tr[R_AX_A]=\gamma$ and $\tr[S_BY_B]=\delta$ are called \emph{primal optimal plan}, respectively \emph{dual optimal plan}.\\
According to the \emph{weak duality theorem} $\gamma\geq\delta$. The difference $\gamma-\delta$ is called \emph{duality gap}. The following theorem called \emph{Slater's condition} establishes an easy-to-check condition under which the duality gap vanishes, that is, $\gamma=\delta$.    
\begin{thm}
\label{slater}
Let $\gamma$ and $\delta$ be defined as above and $(\mathcal{F}, R_A, S_A)$ with $R_A\in\Herm(\Hil_A)$ and $S_B\in\Herm(\Hil_B)$ a semi-definite program. Then the following two implications hold:\\
(i)[Strict dual feasibility] Suppose $\gamma$ is finite and that there exists an operator $Y_B>0$ such that $\mathcal{F}^{\dagger}(Y_B) < R_A$. Then $\gamma=\delta$.\\
(ii) [Strict primal feasibility] Suppose that $\delta$ is finite and that there exists an operator $X_A>0$ such that $\mathcal{F}(X_A) > S_B$. Then $\gamma=\delta$.     
\end{thm}

\section{New Expressions and Bounds for the Smooth Max-Entropy}
\label{sec:hmax-expressions}

In the following, we give alternative expressions for
$H_{\max}(A|B)_{\rho|\sigma}$ and $H_{\max}(A|B)_{\rho}$ based on the
analysis of SDPs. Then, we prove inequalities relating these entropies with a new entropic measure that turns out to be a useful tool for proving the chain rules.


\subsection{New Expressions via SDP Duality}

\begin{lemma}\label{sdpstatement}
Let $\rho_{AB}\in\Su_{\leq}(\Hil_{AB})$, $\sigma_B\in\Su_{\leq}(\Hil_B)$ and let $\rho_{ABC}$ be a purification of $\rho_{AB}$ on an auxiliary Hilbert space $\Hil_C$. Then the max-entropy of A conditioned on B of $\rho_{AB}$ relative to $\sigma_B$ is given by
\begin{equation} \label{presley}
H_{\max}(A|B)_{\rho|\sigma}=\log\min_{ Z_{AB}}\tr[(\id_A\otimes\sigma_B)Z_{AB}],
\end{equation}
where the minimum ranges over all $Z_{AB}\in\Po(\Hil_{AB})$ with $\rho_{ABC}\leq Z_{AB}\otimes\id_C$. 
\end{lemma}
\begin{proof}
Uhlmann's theorem \cite{Uhl76} tells us that the fidelity can be expressed as a maximization of the overlap of purifications in which the optimization goes over one purification only. In particular, if $\rho_{ABC}$ is any purification of $\rho_{AB}$, then by Uhlmann's theorem 
\begin{equation*}
\begin{split}
2^{H_{\max}(A|B)_{\rho|\sigma}}&=F(\rho_{AB}, \id_{A}\otimes\sigma_B)^{2}\\
&=\max_{\substack{X_{ABC}\geq0\\\tr_C[X_{ABC}]=\id_A\otimes\sigma_B\\\rank[X_{ABC}]=1}}F(\rho_{ABC}, X_{ABC})^{2}\\
&\leq\max_{\substack{X_{ABC}\geq0\\\tr_C[X_{ABC}]=\id_A\otimes\sigma_B}}\tr[\rho_{ABC}X_{ABC}]\\
&=\max_{\substack{X_{ABC}\geq0\\\tr_C[X_{ABC}]=\id_A\otimes\sigma_B}}F(\rho_{ABC},X_{ABC})^{2}\\
&\leq F(\rho_{AB}, \id_A\otimes\sigma_B)^{2}\\
&=2^{H_{\max}(A|B)_{\rho|\sigma}},
\end{split}
\end{equation*}
where the first inequality follows from the fact that the set over which we optimize becomes larger and the last inequality follows from the fact that the fidelity is monotonously increasing under the partial trace. The above calculation implies that instead of optimizing over rank one operators $X_{ABC}$ as Uhlmann's theorem demands, one can maximize over all positive semidefinite extensions $X_{ABC}$ of $\id_A\otimes\sigma_B$, that is,  
\begin{equation}
\label{king}
2^{H_{\max}(A|B)_{\rho|\sigma}}=\max_{\substack{X_{ABC}\geq0\\\tr_C[X_{ABC}]=\id_A\otimes\sigma_B}}\tr[\rho_{ABC}X_{ABC}].
\end{equation}
Moreover, for any positive semidefinite operator $X_{ABC}$ with $\tr_C\bigl[X_{ABC}\bigr]\leq\id_A\otimes\sigma_B$ we can define an operator
\[\bar{X}_{ABC}:=X_{ABC}+Y_C\otimes\bigl(\id_A\otimes\sigma_B-\tr_{C}X_{ABC}\bigr),\]    
with $Y_C$ an arbitrary element of $\Su_{=}(\Hil_C)$. By construction it is constrained by $\tr_C\bar{X}_{ABC}=\id_A\otimes\sigma_B$ and also satisfies  
\begin{equation*}
\tr\bigl[\bar{X}_{ABC}\rho_{ABC}\bigr]\geq\tr\bigl[X_{ABC}\rho_{ABC}\bigr]. 
\end{equation*}
Hence, in \eqref{king} it is permissible to take the maximum over the set of all nonnegative operators $X_{ABC}$ whose partial trace $\tr_{C}X_{ABC}$ is bounded by $\id_{A}\otimes\sigma_B$ (in spite of being equal to $\id_{A}\otimes\sigma_B$), that is,
\begin{equation}
\label{kong}
2^{H_{\max}(A|B)_{\rho|\sigma}}=\max_{\substack{X_{ABC}\geq0\\\tr_C[X_{ABC}]\leq\id_A\otimes\sigma_B}}\tr[\rho_{ABC}X_{ABC}].
\end{equation}
Based on \eqref{kong} we can express $2^{H_{\max}(A|B)_{\rho|\sigma}}$ in terms of the following SDP:

{\footnotesize
\vspace{5mm}
\begin{minipage}[b]{0.45\linewidth}\centering
\begin{tabular}{cc}
\hspace{-0.5cm}&\hspace{-0.cm}\textsc{Primal Problem:}\\
\hspace{-0.5cm}minimum:&\hspace{-0.cm}$\tr[(\id_A\otimes\sigma_B)Z_{AB}]$\\
\hspace{-0.5cm}subject to:&\hspace{-0.cm}$Z_{AB}\otimes\id_C\geq\rho_{ABC}$\\
\hspace{-0.5cm}&$Z_{AB}\geq0$.\\
{}&{}\\
{}&{}
\end{tabular}
\end{minipage}
\hspace{0.2cm}
\begin{minipage}[b]{0.45\linewidth}
\centering
\begin{tabular}{cc}
\hspace{-0.3cm}&\textsc{Dual problem}:\\
\hspace{-0.3cm}maximum:&\hspace{-0.2cm}$\tr[X_{ABC}\rho_{ABC}]$\\
\hspace{-0.3cm}subject to:&\hspace{-0.2cm}$\tr_C[X_{ABC}]\leq\id_A\otimes\sigma_B$\\ 
\hspace{-0.3cm}&$X_{ABC}\geq0$\\
{}&{}\\
{}&{}
\end{tabular}
\end{minipage}
\normalsize}
where $Z_{AB}$ is a primal variable and $X_{ABC}$ a dual variable, respectively. Since the space in the dual problem over which one is optimizing, is closed and bounded, it is compact by the Weierstrass theorem. Hence, the dual optimal plan is finite. Furthermore, the operator $\bar{Z}_{AB}=2\|\rho_{ABC}\|_{\infty}\id_{AB}>0$ satisfies Slater's strict primal feasibility condition $2\|\rho_{ABC}\|_{\infty}\id_{ABC}-\rho_{ABC}>0$ and thus the duality gap between the primal and dual optimization problems vanishes. 
\end{proof}


Next, we write out the SDP for $2^{H_{\max}(A|B)_{\rho}}$ and explore the duality gap between the optimization problems. 
\begin{lemma}
\label{maxinflemma}
Let $\rho_{AB}\in\Su_{\leq}(\Hil_{AB})$ and let $\rho_{ABC}$ be a purification of $\rho_{AB}$ on an auxiliary Hilbert space $\Hil_C$. Then 
the max-entropy of A conditioned on B of $\rho_{AB}$ is given by
\begin{equation}
\label{maxiinf}
H_{\max}(A|B)_{\rho}:=\log\min_{Z_{AB}}\|Z_B\|_{\infty},
\end{equation}
where the minimum ranges over all $Z_{AB}\in\Po(\Hil_{AB})$ with $\rho_{ABC}\leq Z_{AB}\otimes\id_C$. 
\end{lemma}
\begin{proof}
  The only thing that changes with respect to the SDP
  in Lemma~\ref{sdpstatement} is that $\sigma_B$ is no longer fixed
  but it becomes a dual variable.  Thus the SDP for
  $2^{H_{\max}(A|B)_{\rho}}$ reads:

{\footnotesize
\vspace{5mm}
\begin{minipage}[b]{0.45\linewidth}\centering
\begin{tabular}{cc}
\hspace{-0.5cm}&\hspace{-0.cm}\textsc{Primal Problem:}\\
\hspace{-0.5cm}minimum:&\hspace{-0.5cm}$\lambda$\\
\hspace{-0.5cm}subject to:&\hspace{-0.5cm}$Z_{AB}\otimes\id_C\geq\rho_{ABC}$\\ 
\hspace{-0.5cm}&$\lambda\id_B\geq\tr_A[Z_{AB}]$\\
\hspace{-0.5cm}&$Z_{AB}\geq0$,~$\lambda\geq0$\\
&
\end{tabular}
\end{minipage}
\hspace{0.2cm}
\begin{minipage}[b]{0.45\linewidth}
\centering
\begin{tabular}{cc}
\hspace{-0.3cm}&\textsc{Dual problem}:\\
\hspace{-0.3cm}maximum:&\hspace{-0.2cm}$\tr[X_{ABC}\rho_{ABC}]$\\
\hspace{-0.3cm}subject to:&\hspace{-0.33cm}$\tr_C[X_{ABC}]\leq\id_A\otimes\sigma_B$\\ 
\hspace{-0.3cm}&$\tr[\sigma_B]\leq1$\\
\hspace{-0.3cm}&$X_{ABC}\geq0$,~$\sigma_B\geq0$\\
&
\end{tabular}
\end{minipage}
\normalsize}
\\
\\
where $\lambda$ and $Z_{AB}$ are primal variables and $\sigma_B$ and
$X_{ABC}$ dual variables. Obviously, the optimal $\lambda$ is equal to
the largest eigenvalue of $Z_B$. Hence, the above program may be
rewritten in the form:

{\footnotesize
\vspace{5mm}
\begin{minipage}[b]{0.45\linewidth}\centering
\begin{tabular}{cc}
\hspace{-0.5cm}&\hspace{-0.cm}\textsc{Primal Problem:}\\
\hspace{-0.5cm}minimum:&\hspace{-0.cm}$\|Z_B\|_{\infty}$\\
\hspace{-0.5cm}subject to:&\hspace{-0.cm}$Z_{AB}\otimes\id_C\geq\rho_{ABC}$\\
\hspace{-0.5cm}&$Z_{AB}\geq0$\\
&
\end{tabular}
\end{minipage}
\hspace{0.2cm}
\begin{minipage}[b]{0.45\linewidth}
\centering
\begin{tabular}{cc}
\hspace{-0.3cm}&\textsc{Dual problem}:\\
\hspace{-0.3cm}maximum:&\hspace{-0.2cm}$\tr[X_{ABC}\rho_{ABC}]$\\
\hspace{-0.3cm}subject to:&\hspace{-0.2cm}$\tr_C[X_{ABC}]\leq\id_A\otimes\sigma_B$\\ 
\hspace{-0.3cm}&$\tr[\sigma_B]\leq1$\\
\hspace{-0.3cm}&$X_{ABC}\geq0,~\sigma_B\geq0$
\end{tabular}
\end{minipage}
\normalsize}
\\
\\
\\
In the dual problem we are optimizing over compact sets, thus there exists a finite dual optimal plan. Furthermore, $\bar{Z}_{AB}=2\|\rho_{ABC}\|_{\infty}\id_{AB}>0$ and $\bar{\lambda}=2\|\bar{Z}_B\|_{\infty}>0$ satisfy Slater's strict primal feasibility condition $\bar{Z}_{AB}\otimes\id_C > \rho_{ABC}$ and $\bar{\lambda}\id_B > \tr_A[\bar{Z}_{AB}]$ which implies a zero duality gap. 
\end{proof}


Note that one can always write the operator norm of $Z_{B}$ as \[\|Z_B\|_{\infty}=\max_{\sigma_B}\tr[\sigma_BZ_B]=\max_{\sigma_B}\tr[(\id_A\otimes\sigma_B)Z_{AB}],\]
where the maximum ranges over all $\sigma_B\in\Su_{\leq}(\Hil_B)$. Expression \eqref{maxiinf} then acquires the form \begin{equation}\label{dakel}H_{\max}(A|B)_{\rho}=\log\min_{\rho_{ABC}\leq Z_{AB}\otimes\id_C}\max_{\sigma_B}\tr[(\id_A\otimes\sigma_B)Z_{AB}].\end{equation} 
On the other hand from the vanishing of the duality gap in the SDP of $H_{\max}(A|B)_{\rho|\sigma}$ it follows that 
\begin{equation*}
\log F(\rho_{AB}, \id_A\otimes\sigma_B)^{2}=\log\min_{Z_{AB}}\tr[(\id_A\otimes\sigma_B)Z_{AB}]
\end{equation*}
which after maximization of the left- and the right-hand sides over $\sigma_B\in\Su_{\leq}(\Hil_B)$ implies 
\begin{equation*}
H_{\max}(A|B)_{\rho}=\log\max_{\sigma_B}\min_{Z_{AB}}\tr[(\id_A\otimes\sigma_B)Z_{AB}].
\end{equation*}
Therefore, the operations $\min$ and $\max$ in \eqref{dakel} commute. Since the function $\tr[(\id_A\otimes\sigma_B)Z_{AB}]$ is bilinear and the sets over which one optimizes are convex, the commutativity of $\min$ and $\max$ can alternatively be seen as a consequence of the minimax theorem.

Henceforth, we will use \eqref{eqmaxentropy}, \eqref{presley} and 
\eqref{maxiinf} and \eqref{dakel} as
interchangeable expressions for the conditional max-entropy and the
conditional relative max-entropy, respectively.

\subsection{A Bound on the Relative Conditional Entropy} 
Here we provide two lemmas which give tight upper bounds of the max- and min-entropy in terms of the relative max- and min-entropy, respectively. The first lemma is a new result whereas the latter one is an improved version of Lemma 21 from \cite{TSSR10}. Both of the following statements are important for the derivation of chain rules.    
\begin{lemma}
\label{supportinglemma}
Let $\eps>0$, $\rho_{AB}\in\Su_{\leq}(\Hil_{AB})$ and $\rho_{AB}^{\prime}\approx_{\eps^{\prime}}\rho_{AB}$. Then there exists a state $\tilde{\rho}_{AB}\approx_{\eps+\eps^{\prime}}\rho_{AB}^{\prime}$ such that 
\begin{equation}
\label{valdes}
H_{\max}(A|B)_{\tilde{\rho}}\leq H_{\max}(A|B)_{\rho|\rho^{\prime}}+\log\left(\frac{1}{1-\sqrt{1-\eps^{2}}}\right).
\end{equation} 
\end{lemma}
\begin{proof}
Let $\tilde{Z}_{AB}$ be an optimal primal plan for the semidefinite program for $H_{\max}(A|B)_{\rho|\rho^{\prime}}$ and $\Pi_B$ be the minimum rank projector onto the smallest eigenvalues of the reduced operator $\tilde{Z}_{B}$ such that $\tr[\Pi_B^{\bot}\rho_B^{\prime}]\leq 1-\sqrt{1-\eps^{2}}$ where $\Pi_B^{\bot}$ is the orthogonal complement of $\Pi_B$ and let $\tilde{\rho}_{AB}:=\Pi_B\rho_{AB}\Pi_B$. By Equation \eqref{maxiinf}, we can write  
\begin{equation*}
\begin{split}
2^{H_{\max}(A|B)_{\tilde{\rho}}}&=\min_{\tilde{\rho}_{ABC}\leq Z_{AB}\otimes\id_C}\|Z_B\|_{\infty}\\
&\leq\|\Pi_B\tilde{Z}_{B}\Pi_B\|_{\infty},
\end{split}
\end{equation*}
where we used the fact that $\rho_{ABC}\leq\tilde{Z}_{AB}\otimes\id_C$ implies $\tilde{\rho}_{ABC}\leq\Pi_B\tilde{Z}_{AB}\Pi_B\otimes\id_C$.
Let $\Pi_B^{\prime}$ be the projector onto the largest eigenvalue of $\Pi_B\tilde{Z}_{B}\Pi_B$. Then the definition of $\Pi_B$ implies that 
\begin{equation}
\label{mac}
\tr[(\Pi_B^{\bot}+\Pi_B^{\prime})\rho_B^{\prime}]\geq1-\sqrt{1-\eps^{2}}.
\end{equation}
Moreover, by construction $\Pi_B^{\bot}$ and $\Pi_B^{\prime}$ project onto orthogonal eigenspaces of $\tilde{Z}_B$, that is,  $\Pi_B^{\bot}\Pi_B^{\prime}=0$. Hence the sum $\Pi_B^{\bot}+\Pi_B^{\prime}$ is itself a projector which commutes with $\tilde{Z}_{B}$. We use the last two facts to find an upper bound for
\begin{equation}
\label{upgrade}
\begin{split}
\|\Pi_B\tilde{Z}_{B}\Pi_B\|_{\infty}&=\tr[\Pi_{B}^{\prime}\tilde{Z}_{B}]\\
&=\min_{\mu_B}\frac{\tr[\mu_B\tilde{Z}_{B}]}{\tr[\mu_B]},
\end{split}
\end{equation} 
where the minimization is over all positive operators in the support of $\Pi_{B}^{\bot}+\Pi_{B}^{\prime}$. Fixing $\mu_B=(\Pi_B^{\bot}+\Pi_B^{\prime})\rho_B^{\prime}(\Pi_B^{\bot}+\Pi_B^{\prime})$, we obtain the following upper bound for \eqref{upgrade}:\\ 
\begin{equation*}
\begin{split}
\|\Pi_B\tilde{Z}_{B}\Pi_B\|_{\infty}&\leq\frac{\tr[(\Pi_B^{\bot}+\Pi_B^{\prime})\rho_B^{\prime}(\Pi_B^{\bot}+\Pi_B^{\prime})\tilde{Z}_{B}]}{\tr[(\Pi_B^{\bot}+\Pi_B^{\prime})\rho_B^{\prime}]}\\
&=\frac{\tr[(\Pi_B^{\bot}+\Pi_B^{\prime})\tilde{Z}_{B}^{1/2}\rho_B^{\prime}\tilde{Z}_{B}^{1/2}]}{\tr[(\Pi_B^{\bot}+\Pi_B^{\prime})\rho_B^{\prime}]}\\
&\leq\frac{\tr[\rho_B^{\prime}\tilde{Z}_{B}]}{\tr[(\Pi_B^{\bot}+\Pi_B^{\prime})\rho_B^{\prime}]}\\
&\leq2^{H_{\max}(A|B)_{\rho|\rho^{\prime}}}\frac{1}{1-\sqrt{1-\eps^{2}}},
\end{split}
\end{equation*}
where in the last line we used Equation \eqref{presley} and Inequality \eqref{mac}. Finally, taking the logarithm on both sides yields \eqref{valdes}.

The proof is concluded by the upper bound  
\begin{equation*}
\begin{split}
&P(\tilde{\rho}_{AB}, \rho_{AB}^{\prime})= P(\Pi_B\rho_{AB}\Pi_B, \rho_{AB}^{\prime})\\
&\leq P(\Pi_B\rho_{AB}\Pi_B, \Pi_B\rho_{AB}^{\prime}\Pi_B)+P(\Pi_B\rho_{AB}^{\prime}\Pi_B, \rho_{AB}^{\prime})\\
&\leq P(\rho_{AB}, \rho_{AB}^{\prime})+\sqrt{2\tr[\Pi_B^{\bot}\rho_{AB}^{\prime}]-(\tr[\Pi_B^{\bot}\rho_{AB}^{\prime}])^{2}}\\
&\leq \eps^{\prime}+\eps 
\end{split}
\end{equation*}
where we use Inequality \eqref{schlumf2} and the fact that the function $\sqrt{2t-t^{2}}$ is monotonously increasing in the interval $[0, 1]$.
\end{proof}

\begin{lemma}
\label{sava}
Let $\eps>0$ and $\rho_{ABC}\in\Su_{\leq}(\Hil_{ABC})$ be pure. Then there exist a projector $\Pi_{AC}$ on $\Hil_{AC}$ and a state $\tilde{\rho}_{ABC}=\Pi_{AC}\rho_{ABC}\Pi_{AC}$ such that $\tilde{\rho}_{ABC}\approx_{\eps}\rho_{ABC}$ and 
\begin{equation*}
H_{\min}(A|B)_{\rho}\leq H_{\min}(A|B)_{\rho|\tilde{\rho}}+\log\left(\frac{1}{1-\sqrt{1-\eps^{2}}}\right).
\end{equation*}
\end{lemma}
As already remarked, the proof of this lemma follows exactly the one of
Lemma 21 in \cite{TSSR10}, up to the following modification. Instead of defining the dual projector $\Pi_B$ of $\Pi_{AC}$ with regard to the pure state $\rho_{ABC}$ such that it satisfies 
\(\tr[\Pi_B^{\bot}\rho_B]\leq \eps^{2}/2\), we demand 
\begin{equation*}
\tr[\Pi_B^{\bot}\rho_B]\leq 1-\sqrt{1-\eps^{2}}. 
\end{equation*}
In this way on the one hand the tighter bound \eqref{schlumf2} yields
\begin{equation*}
\begin{split}
P(\tilde{\rho}_{ABC}, \rho_{ABC})&\leq\sqrt{2\tr[\Pi_{AC}^{\bot}\rho_{ABC}]-(\tr[\Pi_{AC}^{\bot}\rho_{ABC}])^{2}}\\
&=\sqrt{2\tr[\Pi_{B}^{\bot}\rho_{B}]-(\tr[\Pi_B^{\bot}\rho_{B}])^{2}}\\
&\leq\eps
\end{split}
\end{equation*}
which is the same as in Lemma 21 and on the other hand the correction term $\log(2/\eps^{2})$ in Lemma 21 is replaced by the tighter expression $\log(1/1-\sqrt{1-\eps^{2}})$. 

\subsection{The $\eps$-Smooth $S$-Entropy}
For the proof of the chain rules we define an auxiliary entropy measure called $\eps$-smooth 
$S$-entropy\footnote{The idea for this entropy measure was proposed by Robert K\"onig.}.

We assume that $\rho_{AB}\in\Su_{\leq}(\Hil_{AB})$ and $\sigma_B\in\Su_{\leq}(\Hil_B)$ with $\supp(\rho_{B})\subseteq\supp(\sigma_B)$  and denote  for every $\lambda\in\mathbb{R}$ the projector onto the eigenspace corresponding to the strictly negative eigenvalues of the operator $2^{\lambda}\rho_{AB}-\sigma_{B}$ by $P_{AB}^{\lambda}$.
\begin{definition}
\label{def1}
Let $\eps>0$. Then the $\eps$-smooth $S$-entropy of A conditioned on B of $\rho_{AB}$ relative to $\sigma_B$ is defined as 
\begin{equation}
S^{\eps}(A|B)_{\rho|\sigma}:=\inf\bigl\{\lambda\in\mathbb{R}:\tr[P_{AB}^{\lambda}\rho_{AB}]\leq\eps\bigr\}. 
\end{equation}
\end{definition}
\noindent Intuitively, this evaluates in a $\eps$-smoothed way the smallest $\lambda$ for which $\rho_{AB} \geq 2^{-\lambda} \sigma_B$ holds. This should be contrasted with the min-entropy, which
evaluates to the largest $\lambda$ such that $\rho_{AB} \leq 2^{-\lambda} \sigma_B$. The $S$-entropy is a technical tool only, and our results are expressed in terms of the max-entropy instead. In this spirit,
the next lemma gives the upper bound of the  $\eps$-smooth $S$-entropy in terms of the max-entropy. 
\begin{lemma}
\label{lemma2}
Let $\eps>0$, $\rho_{AB}\in\Su_{\leq}(\Hil_{AB})$ and $\sigma_B\in\Su_{\leq}(\Hil_B)$. Then,
\begin{equation}
\label{konradaud}
S^{\eps}(A|B)_{\rho|\sigma}\leq H_{\max}(A|B)_{\rho|\sigma} +\log\left(\frac{1}{\eps^{2}}\right).
\end{equation}
\end{lemma}
\begin{proof}
Let $\lambda_{\inf}\in\mathbb{R}$ be the infimum in Definition \ref{def1}, that is, $\lambda_{\inf}=S^{\eps}(A|B)_{\rho|\sigma}$, let $\lambda=\lambda_{\inf}-\delta$ where $\delta>0$ and let $P_{AB}^{\pm}$ denote the projector onto the nonnegative and strictly negative eigenvalues of $\rho_{AB}-2^{-\lambda}\sigma_B$, respectively. Then, a straightforward computation yields
\begin{align}
2^{\frac{1}{2}H_{\max}(A|B)_{\rho|\sigma}-\frac{1}{2}S^{\eps}(A|B)_{\rho|\sigma}+\frac{1}{2}\delta}&=\|\sqrt{\rho_{AB}}\sqrt{\sigma_{B}}\|_{1}2^{-\frac{1}{2}\lambda}\nonumber\\
&\geq\tr[\sqrt{\rho_{AB}}\sqrt{\sigma_B}]2^{-\frac{1}{2}\lambda}\nonumber\\
&=\tr[\sqrt{\rho_{AB}}\sqrt{2^{-\lambda}\sigma_B}]\nonumber\\
&\geq\tr[P_{AB}^{+}2^{-\lambda}\sigma_B+P_{AB}^{-}\rho_{AB}]\nonumber\\
&\geq\tr[P_{AB}^{-}\rho_{AB}]\nonumber\\
&\geq\eps. \label{sexy}
\end{align}
The first inequality follows from Lemma 9.5 in \cite{nielsen00}. In the fourth line we have applied Corollary \ref{konradcor} and in the last line have used the fact that $P_{AB}^{-}$ is identical with the projector $P_{AB}^{\lambda}$ and $\tr[P_{AB}^{\lambda}\rho_{AB}] \geq \eps$ by definition of $\lambda$ for any $\delta > 0$. Finally, taking the logarithm on both sides of \eqref{sexy} and subsequently taking the limit $\delta\to 0$ we obtain \eqref{konradaud}. 
\end{proof}
\section{Main Results}
\label{sec:main-results}
This section contains the main result of this paper: a derivation of
the previously unknown chain rules for smooth min- and
max-entropies. To simplify presentation hereafter, we introduce the
function 
\begin{align*} 
  f: \, \eps \mapsto \log \frac{1}{1-\sqrt{1-\eps^{2}}}
\end{align*}
 that appears as an error term in the chain rules. It vanishes as $\eps \to 1$ and grows logarithmically in $\frac{1}{\eps}$ when $\eps \to 0$.

As remarked in the introduction, the explicit form of one of the chain rules has already been derived in Lemma A.6 in \cite{dupuis10}. Following the steps of the original proof and using the improved bound from Lemma \ref{sava} we can tighten the chain rule inequality presented in Lemma A.6 of \cite{dupuis10} as follows: 
\begin{thm}
Let $\eps>0$, $\eps^{\prime}$, $\eps^{\prime\prime}\geq0$ and $\rho_{ABC}\in\Su_{\leq}(\Hil_{ABC})$. Then,
\begin{equation*}
H_{\min}^{\eps+\eps^{\prime}+2\eps^{\prime\prime}}(AB|C)_{\rho}\geq H_{\min}^{\eps^{\prime\prime}}(A|BC)_{\rho}+H_{\min}^{\eps^{\prime}}(B|C)_{\rho} -f(\eps).
\end{equation*}
\end{thm}
In the remainder of that section we provide proofs for the remaining three pairs of chain rules. Due to the smooth duality relation \eqref{smoothdual} it is enough to prove only one of each pair.  
\begin{thm}
Let $\eps>0$, $\eps^{\prime}$, $\eps^{\prime\prime}\geq0$ and $\rho_{ABC}\in\Su_{\leq}(\Hil_{ABC})$. Then,
\begin{equation}
\label{cr2}
H_{\min}^{\eps^{\prime}}(AB|C)_{\rho}\leq H_{\min}^{\eps+\eps^{\prime}+2\eps^{\prime\prime}}(A|BC)_{\rho}+H_{\max}^{\eps^{\prime\prime}}(B|C)_{\rho}
+2f(\eps) \,.
\end{equation}
\end{thm}

\begin{proof}
Let $\rho_{ABC}^{\prime}\approx_{\eps^{\prime}}\rho_{ABC}$, $\rho_{BC}^{\prime\prime}\approx_{\eps^{\prime\prime}}\rho_{BC}$ be such that 
\begin{align*}
H_{\min}(AB|C)_{\rho^{\prime}} &= H_{\min}^{\eps^{\prime}}(AB|C)_{\rho}, \quad \textrm{and}\\
H_{\max}(B|C)_{\rho^{\prime\prime}} &= H_{\max}^{\eps^{\prime\prime}}(B|C)_{\rho},
\end{align*}
and let $\sigma_C\in\Su_{\leq}(\Hil_C)$ be such that
\begin{align}
\label{ineq}
\rho_{ABC}^{\prime}&\leq 2^{-H_{\min}(AB|C)_{\rho^{\prime}}}\sigma_{C}=2^{-H_{\min}^{\eps'}(AB|C)_{\rho}}\sigma_C.
\end{align}
For every $\delta>0$ and $\tilde{\eps}>0$ there is a $\delta^{\prime}\in(0, \delta]$ such that the projector $P_{BC}^{\lambda}$ onto the strictly negative eigenvalues of the operator $2^{\lambda}\rho_{BC}^{\prime\prime}-\sigma_C$ with  $\lambda:=S^{\tilde{\eps}}(B|C)_{\rho^{\prime\prime}|\sigma}+\delta^{\prime}$, satisfies the constraint $\tr[P_{BC}^{\lambda}\rho_{BC}^{\prime\prime}]\leq\tilde{\eps}$ in Definition \ref{def1}. If $P_{BC}^{\lambda\bot}$ is the orthogonal complement of $P_{BC}^{\lambda}$, we have
\begin{equation}
\label{operator2}
P_{BC}^{\lambda\bot}\sigma_{C}P_{BC}^{\lambda\bot}\leq 2^{\lambda}P_{BC}^{\lambda\bot}\rho_{BC}^{\prime\prime}P_{BC}^{\lambda\bot}.
\end{equation} 
A conjugation of \eqref{ineq} with $P_{BC}^{\lambda\bot}$ together with \eqref{operator2} yields 
\begin{equation*}
P_{BC}^{\lambda\bot}\rho_{ABC}^{\prime}P_{BC}^{\lambda\bot}\leq2^{-H_{\min}^{\eps^{\prime}}(AB|C)_{\rho}+\lambda }P_{BC}^{\lambda\bot}\rho_{BC}^{\prime\prime}P_{BC}^{\lambda\bot} \,,
\end{equation*}
which is equivalent to 
\begin{equation*}
H_{\min}(A|BC)_{P^{\lambda\bot}\rho^{\prime}P^{\lambda\bot}|P^{\lambda\bot}\rho^{\prime\prime}P^{\lambda\bot}}\geq H_{\min}^{\eps^{\prime}}(AB|C)_{\rho}-\lambda.
\end{equation*}
A subsequent optimization of the left-hand side over all $\Su_{\leq}(\Hil_{BC})$ yields
\begin{equation}
\label{serena}
 H_{\min}(A|BC)_{P^{\lambda\bot}\rho^{\prime}P^{\lambda\bot}}\geq H_{\min}^{\eps^{\prime}}(AB|C)_{\rho}-\lambda
\end{equation}

Since $\rho_{ABC}$ is an extension of $\rho_{BC}$, by Corollary~\ref{cor} there exists an extension $\rho_{ABC}^{\prime\prime}$ of $\rho_{BC}^{\prime\prime}$ such that $P(\rho_{ABC}^{\prime\prime},\rho_{ABC})=P(\rho_{BC}^{\prime\prime},\rho_{BC})$. Then the triangle inequality as well as~\eqref{schlumf1} and~\eqref{schlumf2} give us the following upper bound for the purified distance between $P_{BC}^{\lambda\bot}\rho_{ABC}^{\prime}P_{BC}^{\lambda\bot}$ and $\rho_{ABC}$:    
\begin{multline*}
P(P_{BC}^{\lambda\bot}\rho_{ABC}'P_{BC}^{\lambda\bot},\rho_{ABC})\\
\begin{split}
&\leq P(P^{\lambda\bot}_{BC} \rho_{ABC}' P^{\lambda\bot}_{BC}, P^{\lambda\bot}_{BC}\rho_{ABC}P^{\lambda\bot}_{BC})\\
&+ P(P^{\lambda\bot}_{BC} \rho_{ABC} P^{\lambda\bot}_{BC}, P^{\lambda\bot}_{BC}\rho_{ABC}''P^{\lambda\bot}_{BC})\\
&+ P(P_{BC}^{\lambda\bot}\rho_{ABC}'' P_{BC}^{\lambda\bot}, \rho_{ABC}^{\prime\prime})\\
&+ P(\rho_{ABC}'', \rho_{ABC})\\
&\leq\sqrt{2\tilde{\eps}-\tilde{\eps}^{2}}+\eps^{\prime}+2\eps^{\prime\prime}.
\end{split}
\end{multline*}
After smoothing the left-hand side of \eqref{serena} and upper-bounding the term $S^{\tilde{\eps}}(B|C)_{\rho^{\prime\prime}|\sigma}$ on the right-hand side of \eqref{serena} by $H_{\max}(B|C)_{\rho^{\prime\prime}|\sigma}$ in accordance with Lemma \ref{lemma2} and subsequently optimizing it over $\Su_{\leq}(\Hil_C)$, we obtain  
\begin{equation*}
\begin{split}
H_{\min}^{\eps^{\prime}}(AB|C)_{\rho}&\leq H_{\min}^{\sqrt{2\tilde{\eps}-\tilde{\eps}^{2}}+\eps^{\prime}+2\eps^{\prime\prime}}(A|BC)_{\rho}+H_{\max}^{\eps^{\prime\prime}}(B|C)_{\rho}\\
&+\log\frac{1}{\tilde{\eps}^{2}}+\delta^{\prime}.
\end{split}
\end{equation*}
Finally, the substitution $\tilde{\eps}:=1-\sqrt{1-\eps^{2}}$ leads to the chain rule~\eqref{cr2} in the limit $\delta\longrightarrow0$.
\end{proof}

\begin{thm}
Let $\eps>0$, $\eps^{\prime}$, $\eps^{\prime\prime}\geq0$ and $\rho_{ABC}\in\Su_{\leq}(\Hil_{ABC})$. Then,
\begin{equation}
\label{cr3}
H_{\min}^{\eps^{\prime}}(AB|C)_{\rho}\leq H_{\max}^{\eps^{\prime\prime}}(A|BC)_{\rho}+ H_{\min}^{2\eps+\eps^{\prime}+2\eps^{\prime\prime}}(B|C)_{\rho}
+3f(\eps) \,.
\end{equation}
\end{thm}

\begin{proof}
Let $\rho_{ABCD}$ be a purification of $\rho_{ABC}$. If
\begin{equation*}
H_{\max}^{\eps^{\prime}}(AB|D)_{\rho} \geq H_{\max}^{2\eps+\eps^{\prime}+2\eps^{\prime\prime}}(B|AD)_{\rho}+H_{\min}^{\eps^{\prime\prime}}(A|D)_{\rho} -3f(\eps)
\end{equation*}
holds, then the chain rule follows by the duality relation \eqref{smoothdual}. 
Let $\rho_{ABD}^{\prime}\approx_{\eps^{\prime}}\rho_{ABD}$, $\rho_{AD}^{\prime\prime}\approx_{\eps^{\prime\prime}}\rho_{AD}$ be such that   
\begin{align*}
H_{\max}(AB|D)_{\rho^{\prime}} &= H_{\max}^{\eps^{\prime}}(AB|D)_{\rho} \,, \quad \textrm{and}\\
H_{\min}(A|D)_{\rho^{\prime\prime}}&=H_{\min}^{\eps^{\prime\prime}}(A|D)_{\rho} \,,
\end{align*}
and let $\sigma_D\in\Su_{\leq}(\Hil_D)$ be such that
\begin{equation}
\label{AD}
\rho_{AD}^{\prime\prime} \leq2^{-H_{\min}(A|D)_{\rho^{\prime\prime}}}\sigma_D
 =2^{-H_{\min}^{\eps^{\prime\prime}}(A|D)_{\rho}}\sigma_D \,.
\end{equation}

Again we use the fact that for every $\delta>0$ there exists a $\delta^{\prime}\in(0, \delta]$ such that for $\lambda:=S^{\tilde{\eps}}(AB|D)_{\rho^{\prime}|\sigma}+\delta^{\prime}$, $\tilde{\eps}>0$ , the projector $P_{ABD}^{\lambda}$ onto the strictly negative eigenvalues of the operator $2^{\lambda}\rho_{ABD}^{\prime}-\sigma_D$ satisfies the constraint $\tr[P_{ABD}^{\lambda}\rho_{ABD}^{\prime}]\leq\tilde{\eps}$ in Definition \ref{def1}. If $P_{ABD}^{\lambda\bot}$ denotes the orthogonal complement of $P_{ABD}^{\lambda}$, then
\begin{equation}
\label{ABD}
2^{\lambda}P_{ABD}^{\lambda\bot}\rho_{ABD}^{\prime}P_{ABD}^{\lambda\bot}\geq P_{ABD}^{\lambda\bot}\sigma_DP_{ABD}^{\lambda\bot}.
\end{equation}
A conjugation of \eqref{AD} with $P_{ABD}^{\lambda\bot}$ and a subsequent combination with \eqref{ABD} yields
\begin{equation}
\label{ABD1}
2^{\lambda-H_{\min}^{\eps^{\prime\prime}}(A|D)_{\rho}}P_{ABD}^{\lambda\bot}\rho_{ABD}^{\prime}P_{ABD}^{\lambda\bot}\geq P_{ABD}^{\lambda\bot}\rho_{AD}^{\prime\prime}P_{ABD}^{\lambda\bot}.
\end{equation}

Consider now the max-entropy  
\begin{equation}
\label{star}
2^{H_{\max}(B|AD)_{P^{\lambda\bot}\rho^{\prime}P^{\lambda\bot}|\rho^{\prime\prime}}}=\hspace{-1.2cm}
\mathop{\min_{Z_{ABD} \geq 0}}_{P_{ABD}^{\lambda\bot}\rho_{ABCD}^{\prime}P_{ABD}^{\lambda\bot}\leq Z_{ABD}\otimes\id_C}\hspace{-1.2cm}\tr[(\id_B\otimes\rho_{AD}^{\prime\prime})Z_{ABD}]
\end{equation}
where $\rho_{ABCD}^{\prime}$ is a purification of $\rho_{ABD}^{\prime}$. Making use of \eqref{ABD1} and the inequality
\begin{equation*}
P_{ABD}^{\lambda\bot}\rho_{ABCD}^{\prime}P_{ABD}^{\lambda\bot}\leq P_{ABD}^{\lambda\bot}\otimes\id_C
\end{equation*}
and omitting the identity operator, we can upper-bound the right-hand side of \eqref{star} in the following way: 
\begin{equation*}
\begin{split}
&\leq\tr[\rho_{AD}^{\prime\prime}P_{ABD}^{\lambda\bot}]\\
&\leq2^{\lambda-H_{\min}^{\eps^{\prime\prime}}(A|D)_{\rho}}\tr[P_{ABD}^{\lambda\bot}\rho_{ABD}^{\prime}P_{ABD}^{\lambda\bot}]\\
&\leq2^{\lambda-H_{\min}^{\eps^{\prime\prime}}(A|D)_{\rho}},
\end{split}
\end{equation*} 
where we use that the term $\tr[P_{ABD}^{\lambda\bot}\rho_{ABD}^{\prime}P_{ABD}^{\lambda\bot}]$ is upper bounded by one. Taking the logarithm and substituting $\lambda$ yields
\small
\begin{equation*}
H_{\max}(B|AD)_{P^{\lambda\bot}\rho^{\prime}P^{\lambda\bot}|\rho^{\prime\prime}}\leq S^{\tilde{\eps}}(AB|D)_{\rho^{\prime}|\sigma}+\delta^{\prime}-H_{\min}^{\eps^{\prime\prime}}(A|D)_{\rho}.
\end{equation*}
\normalsize

A subsequent application of Lemma \ref{lemma2} implies 
\begin{align}
H_{\max}(B|AD)_{P^{\lambda\bot}\rho^{\prime}P^{\lambda\bot}|\rho^{\prime\prime}}&\leq H_{\max}(AB|D)_{\rho^{\prime}}-H_{\min}^{\eps^{\prime\prime}}(A|D)_{\rho} \nonumber\\
&+\delta^{\prime}+\log\frac{1}{\tilde{\eps}^{2}} \label{balakov}
\,,
\end{align}
where the max-entropy term on the right-hand side has been optimized on $\Su_{\leq}(\Hil_D)$.
Consider now the left-hand side of \eqref{balakov}. Corollary \ref{cor} guarantees the existence of an extension $\rho_{ABD}^{\prime\prime}$ such that $P(\rho_{AD}^{\prime\prime}, \rho_{AD})=P(\rho_{ABD}^{\prime\prime}, \rho_{ABD})$. Then, it follows that   
\begin{equation*}
\begin{split}
P(P_{ABD}^{\lambda\bot}\rho_{ABD}^{\prime}P_{ABD}^{\lambda\bot}, \rho_{ABD}^{\prime\prime})&\leq P(P_{ABD}^{\lambda\bot}\rho_{ABD}^{\prime}P_{ABD}^{\lambda\bot}, \rho_{ABD}^{\prime})\\
&+P(\rho_{ABD}^{\prime}, \rho_{ABD}^{\prime\prime})\\
&\leq\sqrt{2\tilde{\eps}-\tilde{\eps}^{2}}+\eps^{\prime}+\eps^{\prime\prime}.
\end{split}
\end{equation*}
Thus, according to Lemma \ref{supportinglemma}, there exists a state $\tilde{\rho}_{ABD}\approx_{\eps+\sqrt{2\tilde{\eps}-\tilde{\eps}^{2}}+\eps^{\prime}+2\eps^{\prime\prime}}\rho_{ABD}$ such that   
\begin{equation*}
\begin{split}
H_{\max}(B|AD)_{\tilde{\rho}}&\leq H_{\max}^{\eps^{\prime}}(AB|D)_{\rho}-H_{\min}^{\eps^{\prime\prime}}(A|D)_{\rho}\\
&+\delta^{\prime}+\log \frac{1}{\tilde{\eps}^{2}}+ f(\eps).
\end{split}
\end{equation*}
Smoothing of the left-hand side and regrouping the terms in the last inequality yields   
\begin{equation*}
\label{cr5final}
\begin{split}
H_{\max}^{\eps^{\prime}}(AB|D)_{\rho}&\geq H_{\max}^{\eps+\sqrt{2\tilde{\eps}-\tilde{\eps}^{2}}+\eps^{\prime}+2\eps^{\prime\prime}}(B|AD)_{\rho}+H_{\min}^{\eps^{\prime\prime}}(A|D)_{\rho}\\
&-\delta^{\prime}-\log \frac{1}{\tilde{\eps}^{2}}-f(\eps).
\end{split}
\end{equation*}
Finally, setting $\tilde{\eps}:=1-\sqrt{1-\eps^{2}}$, taking the limit $\delta\rightarrow0$, and applying the duality relation for smooth entropies \eqref{smoothdual}, we obtain chain rule \eqref{cr3}.\\
\end{proof}

The last chain rule follows from chain rule \eqref{cr2} together with Lemma \ref{updown}.
\begin{cor}
Let $\eps^{\prime}$, $\eps^{\prime\prime}$ $\eps^{\prime\prime\prime}\geq0$ and $\rho_{ABC}\in\Su_{\leq}(\Hil_{ABC})$ such that $\eps'+2\eps''+\eps''' < 1 - 2\sqrt{1-\tr\rho}$. Then,
\begin{align}
H_{\min}^{\eps'}(AB|C)_{\rho} &\leq H_{\max}^{\eps'''}(A|BC)_{\rho}+H_{\max}^{\eps''}(B|C)_{\rho} \nonumber\\
 \label{cr4}
 &\quad +g(\eps'\!, \eps''\!\!, \eps'''\!, \tr \rho) \,,
\end{align}
\begin{align*}
&\textrm{where}\ g(\eps'\!, \eps''\!\!, \eps'''\!, \tr \rho) := \\
&\quad \inf_{\eps} \Big\{ 2 f(\eps) + \log \Big( \frac{1}{1-(\eps \!+\! \eps'\! +\! 2\eps'' \!+\! \eps''' \!+\! 2\sqrt{1\!-\!\tr\rho})^{2}} \Big) \Big\},
\end{align*}
and the infimum is taken in the range
$0 < \eps < 1 - \eps' - 2\eps'' - \eps''' - 2\sqrt{1-\tr\rho}$.
\end{cor}

\begin{proof}
Let $\eps > 0$ be any smoothing parameter such that $\eps < 1 - \eps' - 2\eps'' - \eps''' - 2\sqrt{1-\tr\rho}$. Then, by Lemma \ref{updown}, the smooth min-entropy term on the right-hand side of \eqref{cr2} is upper bounded by
\begin{equation*}
H_{\max}^{\eps'''}(A|BC)_{\rho}
+\log \Big( \frac{1}{1-(\eps \!+\! \eps'\! +\! 2\eps'' \!+\! \eps''' \!+\! 2\sqrt{1\!-\!\tr\rho})^{2}} \Big)
\end{equation*}
which immediately gives \eqref{cr4}.
\end{proof}

In contrast to the previous chain rules, the last one leads to non-trivial results even if we apply it to non-smooth entropies. For example, for a normalized state $\rho_{ABC}$, we find
\begin{align*}
  H_{\min}(AB|C)_{\rho} &\leq H_{\max}(A|BC)_{\rho}+H_{\max}(B|C)_{\rho} + 4\,.
\end{align*}

\section{Conclusion}

We derived four pairs of chain rules for the smooth entropy, and every combination of min- and max-entropies is 
considered. Counter-examples suggest that the inequalities cannot be reversed, and thus that this list is complete. In particular, we do not expect a chain rule of the form
\begin{align}
H_{\min}^{\eps}(AB|C) \leq H_{\min}^{\eps'}(A|BC) + H_{\min}^{\eps''}(B|C) + h,
\label{eq:wrong}
\end{align}
for small smoothing parameters $\eps, \eps'$ and $\eps''$ and error term $h(\eps,\eps',\eps'')$ due to the following counter-example.
Let us consider the state $\rho_{ABCC'} = \frac{1}{2} \sum_{i \in \{0,1\}} \rho_{ABC}^i \otimes |i \rangle\!\langle i|_{C'}$ with
\begin{align*}
  \rho_{ABC}^0 = |\phi\rangle\!\langle\phi|_{AB} \otimes \pi_C 
  \quad \textrm{and} \quad \rho_{ABC}^1 = \pi_A \otimes |\phi\rangle\!\langle\phi|_{BC},
\end{align*}
where $|\phi\rangle$ is a maximally entangled state, $\pi$ is a fully mixed state, we take $A$, $B$ and $C$ to be $d$-dimensional quantum systems and $C'$ is an auxiliary register with basis $\{ |0\rangle, |1\rangle \}$. Any min-entropy conditioned on the classical register $C'$ can be expressed as~\cite{M12}
$$H_{\min}(\cdot|\cdot C')_{\rho} = - \log \frac{\sum_{i = 0}^1 2^{-H_{\min}(\cdot|\cdot)_{\rho^i}}}{2}
\approx \min_{i} H_{\min}(\cdot|\cdot)_{\rho^i} ,$$
where we approximate up to $\pm 1$. Thus,
$H_{\min}(AB|CC') = 0$ and $H_{\min}(A|BCC') = H_{\min}(B|CC') \approx -\log d$ and it is easy to verify that~\eqref{eq:wrong} is violated for moderate smoothing $\eps', \eps'' < \frac{1}{2}$ and $d$ such that $\log d \gg h$.

\subsection*{Acknowledgments}

This work was supported by the Swiss National Science Foundation (SNF)
through the National Centre of Competence in Research “Quantum Science
and Technology” and project No. 200020-135048, and by the European
Research Council (ERC) via grant No. 258932. MT acknowledges support
from the National Research Foundation (Singapore), and the Ministry of
Education (Singapore).


%

\appendices

\section{Proof of Lemma~\ref{updown}}
\label{sec:various-proofs}
In the following we restate Lemma \ref{updown} and prove it using the
derived SDPs for the non-smooth max-entropy.
\begin{namedlemma}[Restatement of Lemma \ref{updown}]
Let $\eps$, $\eps^{\prime}\geq0$ such that $\eps+\eps^{\prime}+2\sqrt{1-\tr\rho_{AB}}<1$ and let $\rho_{AB}\in\Su_{\leq}(\Hil_{AB})$. Then,
\begin{equation*}
\begin{split}
H_{\min}^{\eps^{\prime}}(A|B)_{\rho}&\leq H_{\max}^{\eps}(A|B)_{\rho}\\
&+\log\left(\frac{1}{1-(\eps+\eps^{\prime}+2\sqrt{1-\tr\rho})^{2}}\right).
\end{split}
\end{equation*}
\end{namedlemma}

\begin{proof}
Define $\hat{\rho}_{AB}=\rho_{AB}/\tr(\rho_{AB})$. According to Lemma 5.2 in \cite{M12} there are embeddings $U:\Hil_A\longrightarrow\Hil_{A^{\prime}}$ and $V:\Hil_B\longrightarrow\Hil_{B^{\prime}}$  such that there exists a normalized state $\bar{\rho}_{A^{\prime}B^{\prime}}\approx_{\eps}\hat{\rho}_{A^{\prime}B^{\prime}}$, where $\hat{\rho}_{A^{\prime}B^{\prime}}=\left(U\otimes V\right)\hat{\rho}_{AB}\left(U^{\dagger}\otimes V^{\dagger}\right)$, which minimizes the smooth max-entropy $H_{\max}^{\tilde{\eps}}(A^{\prime}|B^{\prime})_{\hat{\rho}}=H_{\max}^{\tilde{\eps}}(A|B)_{\hat{\rho}}$.\\
Consider now the quantity $2^{-H_{\min}^{\tilde{\eps}+\tilde{\eps}^{\prime}}(A^{\prime}|B^{\prime})_{\bar{\rho}}}$. We are simultaneously minimizing over all $\sigma_{B^{\prime}}\in\Su_{\leq}(\Hil_{B^{\prime}})$ and  all states $\tilde{\rho}_{A^{\prime}B^{\prime}}$, that are $\tilde{\eps}+\tilde{\eps}^{\prime}$-close to the normalized state $\bar{\rho}_{A^{\prime}B^{\prime}}$. By Uhlmann's theorem the latter constraint translates into $\tr[\tilde{\rho}_{A^{\prime}B^{\prime}C}\bar{\rho}_{A^{\prime}B^{\prime}C}]\geq1-(\tilde{\eps}+\tilde{\eps}^{\prime})^{2}$ where $\Hil_C$ is a purifying system. We can formulate $2^{-H_{\min}^{\tilde{\eps}+\tilde{\eps}^{\prime}}}(A^{\prime}|B^{\prime})_{\bar{\rho}}$ as the following semidefinite program:
\vspace{5mm}

{\small
\begin{center}
\begin{tabular}{cc}
&\textsc{Primal Problem:}\\
minimum:&\hspace{-0.5cm}$\tr[\id_{B^{\prime}}\sigma_{B^{\prime}}]$\\
subject to:&\hspace{-0.5cm}$\id_{A^{\prime}}\otimes\sigma_{B^{\prime}}\geq\tr_{C}[\tilde{\rho}_{A^{\prime}B^{\prime}C}]$\\ 
&$\tr[\tilde{\rho}_{A^{\prime}B^{\prime}C}\bar{\rho}_{A^{\prime}B^{\prime}C}]\geq 1-(\tilde{\eps}+\tilde{\eps}^{\prime})^{2}$\\
&$\tr[\tilde{\rho}_{A^{\prime}B^{\prime}C}]\leq 1$\\
&$\sigma_{B^{\prime}}\geq0,~\tilde{\rho}_{A^{\prime}B^{\prime}C}\geq0$\\
&
\end{tabular}
\begin{tabular}{cc}
&\textsc{Dual problem}:\\
maximum:&$(1-(\tilde{\eps}+\tilde{\eps}^{\prime})^{2})\lambda-\mu$\\
subject to:&\hspace{-0.3cm}$\tr_A[E_{A^{\prime}B^{\prime}}]\leq\id_B^{\prime}$\\ 
&$\lambda\bar{\rho}_{A^{\prime}B^{\prime}C}\leq E_{A^{\prime}B^{\prime}}\otimes\id_{C}+\mu\id_{A^{\prime}B^{\prime}C}$\\
&$E_{A^{\prime}B^{\prime}}\geq0$,~$\lambda,~\mu\geq0$,\\
\end{tabular}
\end{center}
\normalsize}

\vspace{5mm}
where $\sigma_{B^{\prime}}$ and $\tilde{\rho}_{A^{\prime}B^{\prime}C}$ are the primal variables and $E_{A^{\prime}B^{\prime}}$, $\lambda$ and $\mu$ are the dual variables, respectively. Let $Z_{A^{\prime}B^{\prime}}$ be a primal optimal plan for the semidefinite program of $H_{\max}(A^{\prime}|B^{\prime})_{\bar{\rho}}$, that is $Z_{A^{\prime}B^{\prime}}\otimes\id_C\geq\bar{\rho}_{A^{\prime}B^{\prime}C}$ and $\tr_{A^{\prime}}[Z_{A^{\prime}B^{\prime}}]\leq2^{H_{\max}(A^{\prime}|B^{\prime})_{\bar{\rho}}}\id_{B^{\prime}}$.
Then the variables $E_{A^{\prime}B^{\prime}}=2^{-H_{\max}(A^{\prime}|B^{\prime})_{\bar{\rho}}}Z_{A^{\prime}B^{\prime}}$, $\lambda=2^{-H_{\max}(A^{\prime}|B^{\prime})_{\bar{\rho}}}$ and $\mu=0$ are a dual feasible plan for the above semidefinite program. 
By the weak duality theorem we have then
\begin{equation*}
(1-(\tilde{\eps}+\tilde{\eps}^{\prime})^{2})2^{-H_{\max}(A^{\prime}|B^{\prime})_{\bar{\rho}}}\leq2^{-H_{\min}^{\tilde{\eps}+\tilde{\eps}^{\prime}}(A^{\prime}|B^{\prime})_{\bar{\rho}}}.
\end{equation*}
Taking the logarithm and considering the fact that all states which are $\tilde{\eps}^{\prime}$-close to $\hat{\rho}_{A^{\prime}B^{\prime}}$ are contained in the $(\tilde{\eps}+\tilde{\eps}^{\prime})$-neighborhood of $\bar{\rho}_{A^{\prime}B^{\prime}}$, we get
\begin{equation}
\label{gunter}
\begin{split}
H_{\min}^{\tilde{\eps}^{\prime}}(A^{\prime}|B^{\prime})_{\hat{\rho}}&\leq H_{\min}^{\tilde{\eps}+\tilde{\eps}^{\prime}}(A^{\prime}|B^{\prime})_{\bar{\rho}}\leq H_{\max}^{\tilde{\eps}}(A^{\prime}|B^{\prime})_{\hat{\rho}}\\
&+\log\left(\frac{1}{1-(\tilde{\eps}+\tilde{\eps}^{\prime})^{2}}\right).
\end{split}
\end{equation}
By Proposition 5.3 in \cite{M12} we have $H_{\min}^{\tilde{\eps}^{\prime}}(A^{\prime}|B^{\prime})_{\hat{\rho}}=H_{\min}^{\tilde{\eps}^{\prime}}(A|B)_{\hat{\rho}}$ and $H_{\max}^{\tilde{\eps}}(A^{\prime}|B^{\prime})_{\hat{\rho}}=H_{\max}^{\tilde{\eps}}(A|B)_{\hat{\rho}}$. Finally, substituting in \eqref{gunter} $\tilde{\eps}=\eps+\sqrt{1-\tr(\rho_{AB})}$ and $\tilde{\eps}^{\prime}=\eps^{\prime}+\sqrt{1-\tr(\rho_{AB})}$ and considering that $H_{\min}^{\eps^{\prime}}(A|B)_{\rho}\leq H_{\min}^{\eps^{\prime}+\sqrt{1-\tr\rho}}(A|B)_{\hat{\rho}}$ as well as $H_{\max}^{\eps+\sqrt{1-\tr\rho}}(A|B)_{\hat{\rho}}\leq H_{\max}^{\eps}(A|B)_{\rho}$ we conclude the proof.
\end{proof}

\section{Technical Lemmas}\label{sec:appendix-tech-lemmas}

\subsection{Operator inequalities}
\begin{thm}[\hspace{-0.002cm}\cite{Aud}, Theorem 1]
Let $Q$ and $R$ be positive semidefinite operators on a Hilbert space $\Hil$ and let $0\leq s\leq1$. Then,
\begin{equation}\label{konrad}\tr\left[Q^{s}R^{1-s}\right]\geq\frac{1}{2}\tr\left[Q+R-\left|Q-R\right|\right]\end{equation} 
\end{thm}    
From this theorem we can draw the following useful corollary. 
\begin{cor}
\label{konradcor}
Let $R$ and $Q$ be positive semidefinite operators on a Hilbert space $\Hil$, let $0\leq s\leq1$ and let $P_{+}$ and $P_{-}$ denote the orthogonal projectors onto the eigenspaces corresponding to nonnegative and strictly negative eigenvalues of the operator $Q-R$, respectively. Then,
\[\tr\left[Q^{s}R^{1-s}\right]\geq\tr\left[P_{+}R+P_{-}Q\right]\]
\end{cor}
\begin{proof}
We make the following decomposition of $\left|Q-R\right|$
\begin{equation}
\label{q-r}
\left|Q-R\right|=P_{+}\left(Q-R\right)P_{+}-P_{-}\left(Q-R\right)P_{-},
\end{equation}
where  $P_{\pm}$ are the projectors onto the nonnegative and strictly negative eigenvalues of $Q-R$, respectively. Substituting \eqref{q-r} in \eqref{konrad} and using the fact that $P_{+}+P_{-}=\id$, we obtain   
\begin{flalign*}
\tr\left[Q^{s}R^{1-s}\right]&\geq\frac{1}{2}\tr\left[Q+R-\left|Q-R\right|\right]\\
&=\tr\left[P_{-}Q+\left(\id-P_{-}\right)R\right]\\
&=\tr\left[P_{-}Q+P_{+}R\right].
\end{flalign*} 
\end{proof}
\subsection{Purified Distance: Properties}   
\begin{lemma}[\hspace{-0.002cm}\cite{TCR10}, Lemma 7]\label{cptpmdist}
If $\rho, \sigma\in\Su_{\leq}(\Hil)$ and $\mathcal{E}$ is a trace non-increasing CPM on $\mathcal{L}(\Hil)$, then \begin{equation*}P(\mathcal{E}(\rho), \mathcal{E}(\sigma))\leq P(\rho, \sigma).\end{equation*}
\end{lemma} 
Evidently, for any $0\leq\Pi\leq1$ the map defined by $\rho\longmapsto\Pi\rho\Pi$, $\rho\in\Su_{\leq}(\Hil)$ is a trace non-increasing CPM. Thus, in particular, by the above lemma we have
\begin{equation}
\label{schlumf1}
P(\Pi\rho\Pi, \Pi\sigma\Pi)\leq P(\rho, \sigma)
\end{equation}
for $\rho$, $\sigma\in\Su_{\leq}(\Hil)$. 
\begin{lemma}[\hspace{-0.002cm}\cite{BCCRR11}, Lemma 7] \label{subpur}Let $\rho\in\Su_{\leq}(\Hil)$ and $0\leq\Pi\leq\id$. Then,
\begin{equation*}P(\Pi\rho\Pi, \rho)\leq \frac{1}{\sqrt{\tr\rho}}\sqrt{(\tr\rho)^{2}-(\tr[\Pi^{2}\rho])^{2}}.\end{equation*}
\end{lemma}
When $\Pi$ is a projector, that is $\Pi^{2}=\Pi$, then a straightforward computation yields
\begin{equation}
\label{schlumf2}
P(\Pi\rho\Pi, \rho)\leq\sqrt{2\tr[\Pi^{\bot}\rho]-\left(\tr[\Pi^{\bot}\rho]\right)^{2}}
\end{equation}
where $\Pi^{\bot}=\id-\Pi$ is the orthogonal complement of $\Pi$.
\begin{lemma}[\hspace{-0.002cm}\cite{TCR10}, Lemma 8]
Let $\rho, \sigma\in\Su_{\leq}(\Hil)$, $\Hil^{\prime}\cong\Hil$ and $\bar{\rho}\in\Su_{\leq}(\Hil\otimes\Hil^{\prime})$ be a purification of $\rho$. Then, there exists a purification $\bar{\sigma}\in\Su_{\leq}(\Hil\otimes\Hil^{\prime})$ of $\sigma$ such that $P(\bar{\rho}, \bar{\sigma})=P(\rho, \sigma)$. 
\end{lemma}
From that lemma one infers the following corollary:
\begin{cor}
\label{cor}
Let $\rho, \sigma\in\Su_{\leq}(\Hil)$, $\Hil^{\prime}\cong\Hil$ and $\bar{\rho}\in\Su_{\leq}(\Hil\otimes\Hil^{\prime})$ be an extension of $\rho$. Then, there exists an extension $\bar{\sigma}\in\Su_{\leq}(\Hil\otimes\Hil^{\prime})$ of $\sigma$ such that $P(\bar{\rho}, \bar{\sigma})=P(\rho, \sigma)$. 
\end{cor}

\ifCLASSOPTIONcaptionsoff
  \newpage
\fi

\newpage






\bibliographystyle{arxiv}
\bibliography{bib,shannon}

\begin{thebibliography}{10}

\bibitem{Aud}
K.~Audenaert, J.~Calsamiglia, R.~{ Mu\~{n}oz-Tapia}, E.~Bagen, L.~Masanes,
  A.~Acin, and F.~Verstraete.
\newblock Discriminating states: The quantum {C}hernoff bound.
\newblock {\em Physical Letters Review}, 98:160501--4, 2007.

\bibitem{Bar}
A.~Barvinok.
\newblock {\em A Course in Convexity}, volume~54 of {\em Graduate Studies in
  Mathematics}.
\newblock American Mathematical Sociaty, 2002.

\bibitem{berta08}
M.~Berta.
\newblock Single-shot quantum state merging.
\newblock Master's thesis, ETH Zurich, 2008.
\newblock \texttt{\href{http://arxiv.org/abs/0912.4495}{arXiv:\,0912.4495}}.

\bibitem{BCCRR11}
M.~Berta, M.~Christandl, R.~Colbeck, J.~Rennes, and R.~Renner.
\newblock The uncertainty principle in the presence of quantum memory.
\newblock {\em Nature Physics}, 1734, 2010.

\bibitem{Dat09}
N.~Datta.
\newblock Min- and max-relative entropies and a new entanglement monotone.
\newblock {\em IEEE Transactions on Information Theory}, 55(6):2816, 2009.

\bibitem{RARDV11}
L.~del Rio, J.~Aberg, R.~Renner, O.~C.~O. Dahlsten, and V.~Vedral.
\newblock The thermodynamic meaning of negative entropy.
\newblock {\em Nature}, 474(7349):61--63, 2011.

\bibitem{dupuis09}
F.~Dupuis.
\newblock {\em The Decoupling Approach to Quantum Information Theory}.
\newblock PhD thesis, Universit\'{e} de Montr\'{e}al, Apr. 2009.
\newblock \texttt{\href{http://arxiv.org/abs/1004.1641}{arXiv:\,1004.1641}}.

\bibitem{dupuis10}
F.~Dupuis, M.~Berta, J.~Wullschleger, and R.~Renner.
\newblock The decoupling theorem.
\newblock Dec. 2010.
\newblock \texttt{\href{http://arxiv.org/abs/1012.6044}{arXiv:\,1012.6044}}.

\bibitem{KRS09}
R.~K\"{o}nig, R.~Renner, and C.~Schaffner.
\newblock The operational meaning of min- and max-entropy.
\newblock {\em IEEE Transactions on Information Theory}, 55(9):4674--4681,
  2009.

\bibitem{nielsen00}
M.~A. Nielsen and I.~Chuang.
\newblock {\em {Quantum Computation and Quantum Information}}.
\newblock Cambridge University Press, 2000.

\bibitem{Ren05}
R.~Renner.
\newblock {\em Security of Quantum Key Distribution}.
\newblock PhD thesis, ETH Z\"{u}rich, 2005.
\newblock Available online: \url{http://arxiv.org/abs/quant-ph/0512258v2}.

\bibitem{RK05}
R.~Renner and R.~K\"{o}nig.
\newblock Universally composable privacy amplification against quantum
  adversaries.
\newblock {\em Springer Lecture Notes in Computer Science}, 3378(9):407--425,
  2005.

\bibitem{RW04}
R.~Renner and S.~Wolf.
\newblock Smooth {R}\'enyi entropy and applications.
\newblock {\em Proc. IEEE Int. Symp. Info. Theory}, page 233, 2004.

\bibitem{shannon48}
C.~Shannon.
\newblock A mathematical theory of communication.
\newblock {\em Bell Syst. Tech. J.}, 27:379--423, 1948.

\bibitem{M12}
M.~Tomamichel.
\newblock {\em A Framework for Non-Asymptotic Quantum Information Theory}.
\newblock PhD thesis, ETH Z\"{u}rich, 2012.
\newblock Available online: \url{http://arxiv.org/abs/arXiv:1203.2142}.

\bibitem{TCR09}
M.~Tomamichel, R.~Colbeck, and R.~Renner.
\newblock A fully quantum asymptotic equipartition property.
\newblock {\em IEEE Transactions on Information Theory}, 55:5840--5847, 2009.

\bibitem{TCR10}
M.~Tomamichel, R.~Colbeck, and R.~Renner.
\newblock Duality between smooth min- and max-entropies.
\newblock {\em IEEE Transactions on Information Theory}, 56:4674--4681, 2010.

\bibitem{TSSR10}
M.~Tomamichel, R.~Renner, C.~Schaffner, and A.~Smith.
\newblock Leftover hashing against quantum side information.
\newblock {\em Proc. IEEE Int. Symp. Info. Theory}, pages 2703--2707, 2010.

\bibitem{Uhl76}
A.~Uhlmann.
\newblock The transition probability in the state space of a \mbox{*-a}lgebra.
\newblock {\em Rep. Math. Phys.}, 9(273), 1976.

\bibitem{JW}
J.~Watrous.
\newblock Theory of quantum information, Fall 2011.
\newblock Available online: \url{http://www.cs.uwaterloo.ca/~watrous/CS766/}.
\newblock Lecture notes.

\bibitem{WTHR11}
S.~Winkler, M.~Tomamichel, S.~Hengl, and R.~Renner.
\newblock Impossibility of growing quantum bit commitments.
\newblock {\em Phys. Rev. Lett.}, 107:090502, 2011.

\end{thebibliography}

\end{document}